\documentclass[11pt,reqno]{amsart}
\usepackage[T1]{fontenc}
\usepackage{amsmath,amssymb,amsthm,mathtools,mathrsfs}
\usepackage{newtxtext,newtxmath}
\usepackage{microtype,booktabs,array,enumitem}
\usepackage{xcolor}
\usepackage[colorlinks=true,linkcolor=black,citecolor=black,urlcolor=blue!45!black]{hyperref}
\usepackage[margin=1.00in]{geometry}
\hypersetup{
 pdftitle={The Lattice of Coalition Logics between MCL and CL: Coatoms and Compact Meets},
 pdfauthor={Shanxia Wang and Yuqiang Dou},
 pdfsubject={The MCL-CL lattice of intermediate coalition logics: coatoms, relative covers, and compact meets},
 pdfkeywords={lattice of coalition logics, intermediate coalition logics, extension lattice, coatoms, compact meets, critical coalition algebras}
}
\allowdisplaybreaks[2]
\newtheorem{theorem}{Theorem}[section]
\newtheorem{lemma}[theorem]{Lemma}
\newtheorem{proposition}[theorem]{Proposition}
\newtheorem{corollary}[theorem]{Corollary}

\theoremstyle{definition}
\newtheorem{definition}[theorem]{Definition}
\newtheorem{construction}[theorem]{Construction}
\theoremstyle{remark}
\newtheorem{remark}[theorem]{Remark}
\newtheorem{problem}[theorem]{Problem}
\DeclareMathOperator{\Cn}{Cn}
\DeclareMathOperator{\Th}{Th}
\DeclareMathOperator{\Var}{Var}
\DeclareMathOperator{\Coat}{Coat}
\DeclareMathOperator{\Log}{Log}
\newcommand{\MCL}{\ensuremath{\mathsf{MCL}}}
\newcommand{\CL}{\ensuremath{\mathsf{CL}}}
\newcommand{\Ser}{\ensuremath{\mathsf{Ser}}}
\newcommand{\IA}{\ensuremath{\mathsf{IA}}}
\newcommand{\Det}{\ensuremath{\mathsf{Det}}}

\newcommand{\Vmcl}{\mathcal V_{\MCL}}
\newcommand{\Vcl}{\mathcal V_{\CL}}
\newcommand{\Fm}{\mathrm{Fm}}
\newcommand{\setP}{\mathcal P}
\providecommand{\llbracket}{\mathopen{[\![}}
\providecommand{\rrbracket}{\mathclose{]\!]}}
\newcommand{\lattice}{\mathscr L_N}
\DeclareMathOperator{\Hist}{Hist}
\DeclareMathOperator{\last}{last}
\title[Lattice of coalition logics]{The Lattice of Coalition Logics between MCL and CL:\\ Coatoms and Compact Meets}
\author{Shanxia Wang}
\address{School of Computer and Information Engineering (School of Artificial Intelligence),
Henan Normal University, Xinxiang, Henan, China}
\email{wangshanxia@htu.edu.cn}
\author{Yuqiang Dou}
\address{School of Computer and Information Engineering (School of Artificial Intelligence),
Henan Normal University, Xinxiang, Henan, China}
\thanks{Both authors: School of Computer and Information Engineering
(School of Artificial Intelligence), Henan Normal University,
Xinxiang, Henan, China. Corresponding author: Shanxia Wang
(\texttt{wangshanxia@htu.edu.cn}).}
\date{October 2026}
\begin{document}

\begin{abstract}
	For a fixed finite set $N$ of at least two agents, we investigate the
	lattice $\lattice$ of coalition logics between minimal coalition logic
	(\MCL) and classical coalition logic (\CL), with extensions required to
	be structural. This lattice is order-dual to an interval in the
	subvariety lattice of Boolean coalition algebras. The eight basic
	extensions obtained by imposing seriality, independence of agents,
	and determinism form a Boolean cube under inclusion, but not a
	sublattice of $\lattice$. Our main results concern coatoms and compact meets.
	First, finite critical coalition algebras, realized as complex
	algebras of finite general concurrent games, yield three infinite
	families of pairwise distinct coatoms containing seriality,
	distinguished by failures of independence, determinism, or both. An
	explicit splitting characterizes the unique coatom omitting
	seriality. Second, the meet of two basic extensions is compact if and
	only if the extensions are comparable. For incomparable pairs, no
	axiomatization of the meet can use formulas of uniformly bounded
	modal depth, even with infinitely many axioms. Consequently,
	$\lattice$ and three natural principal intervals are non-coherent
	algebraic lattices. These results reveal lattice-theoretic phenomena
	beyond the inclusion order of the eight basic extensions.
\end{abstract}

\maketitle
\par
\smallskip

\section{Introduction}

Coalition logic formalizes the ability of groups of agents to enforce
outcomes through their joint choices in games. Pauly's classical
coalition logic (\CL) provides a foundational account of such
abilities \cite{Pauly2002}. Minimal coalition logic (\MCL), introduced
by Li and Ju, allows general concurrent games without imposing
seriality of available actions, independence of agents' choices, or
determinism of complete action profiles \cite{LiJu2025}. Imposing
selected combinations of these conditions yields eight basic logics,
for which complete axiomatizations are available \cite{LiJu2026};
equivalent semantic presentations have also been investigated
\cite{ChenJu2026}. These developments motivate a complementary
question: how are the resulting logics situated within the lattice of
structural extensions? Here an extension is called \emph{structural}
when it preserves all the inference rules of \MCL---modus ponens (MP),
uniform substitution (US), and conditional necessitation (CN)---
rather than merely being closed under uniform substitution.

Fix a finite set $N$ of agents with $|N|\ge2$ and consider
\[
\lattice
=\{L:\MCL\subseteq L\subseteq\CL,\
L\text{ is closed under MP, US, and CN}\},
\]
ordered by inclusion of theorem sets. We call this the \emph{lattice
	of intermediate coalition logics}, or the \emph{lattice of coalition
	logics between \MCL\ and \CL}. As shown below, $\lattice$ is a
complete algebraic lattice. The specified bounds matter: our object
is the interval of structural extensions between these two systems,
not the collection of all logics that could be called coalition
logics. The study of extension lattices has been fruitful in normal
modal logic \cite{Blok1980,Kracht1999}; here it leads to questions
about maximal proper extensions, relative covers, compact elements,
and lattice operations in a strategic setting.

The distinction between the basic systems and their ambient lattice
is already visible at the level of meets. The extensions $B_\Sigma$
obtained by adding any subset $\Sigma\subseteq\{\Ser,\IA,\Det\}$
form a Boolean cube \emph{as an ordered set}, but not as a sublattice
of $\lattice$. In particular,
\[
\MCL\subsetneq
B_{\{\Ser\}}\cap B_{\{\IA\}},
\]
as established in Proposition~\ref{prop:non-sublattice}. Thus
intersecting the sets of imposed strategic assumptions does not, in
general, compute the meet of the corresponding logics. The
surrounding extension lattice has structure that cannot be recovered
from the eight-element order diagram alone.

Our first group of results concerns the \emph{coatoms} of $\lattice$,
that is, its maximal proper elements below \CL. We establish an
order-reversing correspondence with the interval of varieties
$[\Vcl,\Vmcl]$ of Boolean coalition algebras, under which coatoms
correspond to covers of $\Vcl$
(Proposition~\ref{prop:duality}). Using finite critical coalition
algebras realized as complex algebras of explicit general concurrent
games, we obtain three infinite families of pairwise distinct coatoms
containing seriality. These families distinguish failures of
independence, determinism, and both conditions
(Theorem~\ref{thm:families}). In contrast, the nonserial boundary is
rigid: there is a unique coatom $C_{\mathrm{ser}}$ omitting seriality,
and it satisfies
\[
L\subseteq C_{\mathrm{ser}}
\quad\Longleftrightarrow\quad
[\varnothing]\top\notin L
\qquad(L\in\lattice).
\]
This splitting characterization is proved in
Theorem~\ref{thm:serial-splitting}.

Our second group of results concerns \emph{compact meets}. In
$\lattice$, compactness is equivalent to finite axiomatizability
as a structural extension of \MCL{}
(Proposition~\ref{prop:duality}). For the eight basic extensions,
we prove the exact criterion
\[
B_\Sigma\cap B_\Theta\text{ is compact in }\lattice
\quad\Longleftrightarrow\quad
\Sigma\subseteq\Theta\ \text{or}\ \Theta\subseteq\Sigma
\]
(Theorem~\ref{thm:compact-meet}). If the two extensions are
incomparable, their meet cannot be axiomatized even by an infinite
set of formulas whose modal depths have a common finite bound
(Corollary~\ref{cor:unbounded-depth}). Consequently, $\lattice$
and the principal intervals $[S,\CL]$, $[I,\CL]$, and $[D,\CL]$
are non-coherent algebraic lattices
(Corollary~\ref{cor:noncoherent}). The negative results are proved
by separating two distinct local strategic defects through
arbitrarily long sequences of classical play.

The techniques draw on established work in modal extension theory
and universal algebra. In particular, the variety-theoretic
arguments use congruence distributivity and J\'onsson's lemma
\cite{Jonsson1967,Blok1980,Kracht1999}; the cyclic Boolean core of
our finite witnesses has precedents in the theory of rotational
lattices \cite{CzedliNagy2013}. Full-action-history unfoldings for
general concurrent games have also been studied
\cite{ChenJu2026}. The contribution here is the coalition-specific
realization of finite relative covers, the explicit splitting of the
nonserial boundary, and the exact compact-meet criterion for the
eight basic extensions. The algebraic correspondence covers
structural intermediate logics without assuming that each such logic
is complete with respect to a class of general concurrent game
frames; likewise, finite critical witnesses do not by themselves
establish finite axiomatizability of the corresponding coatoms.

The remainder of the paper introduces the language and game
semantics (Section~\ref{sec:sem}) and the algebraic representation
of $\lattice$ (Section~\ref{sec:algebra}). It then constructs
serial coatoms from finite critical algebras
(Section~\ref{sec:finite-coatoms}) and establishes the nonserial
splitting result (Section~\ref{sec:splitting}). Compact meets are
studied in Section~\ref{sec:meets}, followed by further structural
questions about the lattice of coalition logics
(Section~\ref{sec:scope}).

\section{The fixed language and game semantics}\label{sec:sem}

Fix throughout a finite set $N$ of agents with $|N|\ge2$, and a
countably infinite set of propositional variables. The language is
formed from propositional variables by Boolean connectives and the
modalities $[C]$, one for each coalition $C\subseteq N$. The formula
$[C]\varphi$ asserts that $C$ has an \emph{available joint action}
that guarantees $\varphi$. Thus $[C]$ expresses existential strategic
ability, rather than an ordinary universal box. For a singleton
coalition we write $[a]$ instead of $[\{a\}]$.

\begin{definition}\label{def:systems}
	The system \MCL\ contains all propositional tautologies and the
	following axiom schemes, where $C,D\subseteq N$:
	\begin{align}
		[\varnothing](\varphi\to\psi)&\to([C]\varphi\to[C]\psi)
		&&\textup{(MG)},\label{eq:MG}\\
		[C]\varphi&\to[D]\varphi
		&&(C\subseteq D)\quad\textup{(MC)},\label{eq:MC}\\
		\neg[C]\bot&&&\textup{(Live)}.\label{eq:Live}
	\end{align}
	Its inference rules are modus ponens (MP), uniform substitution (US),
	and conditional necessitation
	\begin{equation}\label{eq:CN}
		\frac{\varphi}{[C]\psi\to[\varnothing]\varphi}
		\qquad\textup{(CN)}.
	\end{equation}
	Classical coalition logic \CL\ is obtained from \MCL\ by adding
	\begin{align}
		[C]\top&&&\textup{(Ser)},\label{eq:Ser}\\
		([C]\varphi\wedge[D]\psi)&\to[C\cup D](\varphi\wedge\psi)
		&&(C\cap D=\varnothing)\quad\textup{(IA)},\label{eq:IA}\\
		[C](\varphi\vee\psi)&\to([C]\varphi\vee[N]\psi)
		&&\textup{(Det)}.\label{eq:Det}
	\end{align}
	For a set of formulas $\Gamma$, let $\Cn(\Gamma)$ denote the least
	extension of \MCL\ containing $\Gamma$ and closed under MP, US, and
	CN. Members of $\Gamma$ are added as axioms, and CN applies to
	theorems of the resulting extension. In particular, $\Cn$ describes
	closure of theorem sets, not a local consequence relation permitting
	CN on undischarged assumptions.
\end{definition}

We use the full-profile-first presentation of general concurrent
games \cite{LiJu2025}. A frame is a structure
$F=(W,U,(O_s)_{s\in W})$, where $W$ is a nonempty set of states,
$U$ is a nonempty common action alphabet, and
$O_s:U^N\to\setP(W)$ assigns an outcome set to every full action
profile $\pi\in U^N$ at $s$. The profile $\pi$ is \emph{available}
at $s$ precisely when $O_s(\pi)\ne\varnothing$. For a coalition
$C\subseteq N$ and a partial profile $\sigma\in U^C$, put
\[
O_s^C(\sigma)
=\bigcup\{O_s(\pi):\pi\in U^N,\ \pi|_C=\sigma\}.
\]
A partial profile is available iff $O_s^C(\sigma)\ne\varnothing$.
This includes the empty coalition, whose unique partial profile is
the empty function.

A model $\mathcal M=(F,V)$ adds a valuation $V$ of propositional
variables by subsets of $W$. Writing
$\llbracket\varphi\rrbracket_{\mathcal M}
=\{t\in W:\mathcal M,t\models\varphi\}$, the ability clause is
\begin{equation}\label{eq:semantics}
	\mathcal M,s\models[C]\varphi
	\quad\Longleftrightarrow\quad
	\exists\sigma\in U^C\,\bigl(
	O_s^C(\sigma)\ne\varnothing\ \text{and}\
	O_s^C(\sigma)\subseteq\llbracket\varphi\rrbracket_{\mathcal M}
	\bigr).
\end{equation}
We write $F\models\varphi$ if $\varphi$ holds at every state
under every valuation on $F$, and $F\models\Gamma$ if this holds
for every $\varphi\in\Gamma$.

The three additional conditions on general concurrent game frames
are imposed at every state $s$:
\begin{enumerate}[label=\textup{(\roman*)},leftmargin=2.5em,itemsep=2pt]
	\item \emph{Seriality:} $O_s(\pi)\ne\varnothing$ for some
	$\pi\in U^N$.
	\item \emph{Independence of agents:} whenever $C\cap D=\varnothing$
	and $\sigma\in U^C$, $\tau\in U^D$ are available,
	their union $\sigma\cup\tau$ is available for $C\cup D$;
	equivalently,
	$O_s^{C\cup D}(\sigma\cup\tau)\ne\varnothing$.
	\item \emph{Determinism:} each nonempty $O_s(\pi)$ is a singleton.
\end{enumerate}
These conditions validate, respectively, the Ser, IA, and Det axiom
schemes. We rely on the established soundness and completeness
results for the relevant classes of games
\cite{Pauly2002,LiJu2025,LiJu2026}. Importantly, validity of the IA
or Det scheme on an individual frame need not imply that the
corresponding action-level condition holds on that frame
\cite{LiJu2026}. Accordingly, whenever we claim that a scheme fails
on a constructed frame, we give a refuting formula instance and
valuation, rather than inferring invalidity from the failure of an
action condition alone.

Define the successor relation $R$ of $F$ by
\[
sRt\quad\Longleftrightarrow\quad
t\in O_s(\pi)\text{ for some }\pi\in U^N.
\]
The subframe generated by $s$ consists of all states reachable from
$s$ by finitely many $R$-steps, with outcome functions restricted
to those states. Since every outcome of a retained state is also
retained, restricting a model to a generated subframe preserves the
truth of every formula at its states, by induction on formulas. We
call a local game \emph{classical} if it satisfies all three
action-level conditions above.

The following derived modality, also considered in
\cite{LiJu2025}, will be central to our analysis:
\begin{equation}\label{eq:derived-box}
	e:=[\varnothing]\top,\qquad
	\Box\varphi:=e\to[\varnothing]\varphi.
\end{equation}
At a state $s$, $e$ holds exactly when at least one full profile
is available, or equivalently when $R(s)\ne\varnothing$.
The semantics of the empty coalition therefore yields
\begin{equation}\label{eq:boxsem}
	\mathcal M,s\models\Box\varphi
	\quad\Longleftrightarrow\quad
	\forall t\,(sRt\Rightarrow\mathcal M,t\models\varphi).
\end{equation}
In particular, $\Box\varphi$ is true at a dead-end state. Unlike
the strategic-ability operators $[C]$, the derived $\Box$ is a
normal modality already over \MCL. Indeed, if $\varphi$ is an
\MCL-theorem, CN with $C=\varnothing$ and $\psi=\top$ gives
$e\to[\varnothing]\varphi$, proving the derived necessitation rule
$\varphi/\Box\varphi$. The same rule holds in every structural
extension of \MCL, since CN remains available. Furthermore, MG
with $C=\varnothing$ gives
\[
[\varnothing](\varphi\to\psi)
\to\bigl([\varnothing]\varphi\to[\varnothing]\psi\bigr).
\]
Propositional reasoning, placing the same antecedent $e$ before
each of these three $[\varnothing]$-formulas, yields the K axiom
\[
\Box(\varphi\to\psi)
\to(\Box\varphi\to\Box\psi).
\]
Also, CN applied to the theorem $\top$ yields
$[C]\varphi\to e$, while MC gives $e\to[C]\top$.
Consequently,
\[
\MCL\vdash[C]\varphi\to e,
\qquad
\MCL\vdash[C]\top\leftrightarrow e
\quad(C\subseteq N).
\]
Thus adjoining Ser is equivalent to adjoining the single
variable-free formula $e$. Modal depth $\operatorname{md}(\varphi)$
is measured in the original $[C]$-language; in particular, after
expanding the abbreviation $\Box\psi$, we have
$\operatorname{md}(\Box\psi)=1+\operatorname{md}(\psi)$.

\paragraph{The role of conditional necessitation.}
The requirement of CN closure is equivalent to necessitation for the
normal modality just defined. More precisely, let $L$ contain \MCL\
and be closed under MP and US, without initially assuming CN closure.
Then
\[
  L\text{ is closed under CN}
  \quad\Longleftrightarrow\quad
  \bigl(\varphi\in L\Longrightarrow\Box\varphi\in L\bigr)
  \quad\text{for every formula }\varphi.
\]
For the forward implication, apply CN with $C=\varnothing$ and
$\psi=\top$. Conversely, if $\varphi\in L$ implies
$\Box\varphi=e\to[\varnothing]\varphi\in L$, combine this
with the \MCL-theorem $[C]\psi\to e$ to obtain
$[C]\psi\to[\varnothing]\varphi\in L$ for every $C,\psi$.
Thus CN is the original theorem-level inference rule of \MCL,
expressed as necessitation for the derived box; it is not an
additional axiom and does not impose seriality. In contrast, the
unconditional rule $\varphi/[C]\varphi$ would yield $[C]\top$
from $\top$, excluding dead-end states. Moreover, CN and MG give
\[
  \varphi\to\psi\in L
  \quad\Longrightarrow\quad
  [C]\varphi\to[C]\psi\in L
\]
for every CN-closed $L$. Hence provable equivalence is preserved by
all strategic operators, a fact used in the variety duality of
Proposition~\ref{prop:duality}.

\begin{lemma}[Fixed-frame soundness]\label{lem:frame-soundness}
	For every fixed general concurrent game frame $F$, the set
	$\Log(F)$ of formulas valid on $F$ under all valuations contains
	\MCL\ and is closed under MP, US, and CN. Hence
	\[
	F\models\Gamma\quad\Longrightarrow\quad
	F\models\Cn(\Gamma).
	\]
\end{lemma}
\begin{proof}
	The \MCL\ axioms are valid under the ability semantics:
	nonempty outcomes give Live, extension of an available partial
	profile gives MC, and the truth of the empty-coalition premise
	in MG guarantees the relevant implication at every available
	outcome. MP preserves validity. For US, assign to each variable
	the truth set of the formula substituted for it; the substitution
	lemma, proved by induction on formulas, gives validity of every
	substitution instance.
	
	For CN, suppose $F\models\varphi$ and fix an arbitrary valuation
	and state $s$. If $[C]\psi$ is false at $s$, the CN conclusion is
	true. Otherwise, some available $C$-profile has a nonempty outcome
	set and hence has an available full-profile extension. Thus the
	unique empty-coalition profile is available at $s$. Since
	$\varphi$ is valid at every state under every valuation on $F$,
	all outcomes of \emph{every} available full profile satisfy
	$\varphi$. Therefore $[\varnothing]\varphi$ holds at $s$, as
	required. Closure under the three rules proves the final assertion.
\end{proof}

\section{Algebraic representation and finite relative covers}\label{sec:algebra}

We first represent the coalition-logic extension lattice $\lattice$
by an interval of subvarieties. This algebraic representation applies
to all structural intermediate logics, independently of whether they
admit a completeness theorem for a particular class of games.

A \emph{Boolean coalition algebra} is an expansion
\[
\mathbf A=(B,\wedge,\vee,\neg,0,1,(f_C)_{C\subseteq N})
\]
of a Boolean algebra by unary operations indexed by coalitions.
Let $\Vmcl$ be the variety determined by the Boolean equations and
\begin{align}
	f_C(0)&=0,
	&f_C(x)&\le f_D(x)\quad(C\subseteq D),\label{eq:alg1}\\
	f_\varnothing(\neg x\vee y)\wedge f_C(x)&\le f_C(y),
	&f_C(x)&\le f_\varnothing(1).\label{eq:alg2}
\end{align}
As usual, an inequality $s\le t$ abbreviates the Boolean equation
$s\wedge t\approx s$. The variety $\Vcl\subseteq\Vmcl$ is defined
by additionally requiring the algebraic counterparts of Ser, IA,
and Det:
\begin{align}
	f_C(1)&=1,\label{eq:alg-ser}\\
	f_C(x)\wedge f_D(y)&\le f_{C\cup D}(x\wedge y)
	\quad(C\cap D=\varnothing),\label{eq:alg-ia}\\
	f_C(x\vee y)&\le f_C(x)\vee f_N(y).\label{eq:alg-det}
\end{align}
All these identities are understood universally, for the indicated
coalitions and all algebra elements.

\begin{proposition}[Variety duality for the coalition-logic lattice]
	\label{prop:duality}
	The lattice $\lattice$ is order-dual to the interval of varieties
	$[\Vcl,\Vmcl]$. It is a complete algebraic lattice, and for every
	family $(L_j)_{j\in J}\subseteq\lattice$, including the empty family,
	\[
	\bigwedge_{j\in J}L_j
	=\CL\cap\bigcap_{j\in J}L_j,
	\qquad
	\bigvee_{j\in J}L_j
	=\Cn\!\left(\bigcup_{j\in J}L_j\right).
	\]
	Here $\bigcap\varnothing=\Fm_N$, the set of all formulas, so
	$\bigwedge\varnothing=\CL$ and $\bigvee\varnothing=\MCL$.
	An element $L\in\lattice$ is compact if and only if
	$L=\Cn(\Gamma)$ for some finite set $\Gamma\subseteq L$.
\end{proposition}
\begin{proof}
	For a structural intermediate logic $L$, define
	$\varphi\equiv_L\psi$ if $\varphi\leftrightarrow\psi\in L$.
	This is a Boolean congruence. To check compatibility with every
	strategic operator, suppose $\varphi\to\psi\in L$.
	By CN,
	\[
	[C]\varphi\to[\varnothing](\varphi\to\psi)\in L.
	\]
	Combining this with MG and propositional reasoning gives
	\[
	[C]\varphi\to[C]\psi\in L.
	\]
	Applying the same argument in the reverse direction shows that
	$\varphi\equiv_L\psi$ implies
	$[C]\varphi\equiv_L[C]\psi$. Thus $\equiv_L$ is a congruence
	for the entire signature, and the Lindenbaum--Tarski quotient
	$\mathbf F_L$ is a well-defined Boolean coalition algebra.
	
	The defining axioms of \MCL\ yield
	\eqref{eq:alg1}--\eqref{eq:alg2} in $\mathbf F_L$.
	In particular, CN applied to the theorem $\top$ gives
	$[C]\varphi\to[\varnothing]\top$, corresponding to
	$f_C(x)\le f_\varnothing(1)$. Conversely, the equations defining
	$\Vmcl$ validate every \MCL\ axiom scheme. They also make CN
	sound for \emph{global equational validity}: if $\varphi$ evaluates
	to $1$ under every valuation in an algebra, then
	\[
	f_C(\llbracket\psi\rrbracket)
	\le f_\varnothing(1)
	=f_\varnothing(\llbracket\varphi\rrbracket),
	\]
	which is precisely the algebraic validity of the CN conclusion.
	MP and US preserve global validity as well.
	
	For completeness of the correspondence, observe that every valuation
	of the propositional variables into $\mathbf F_L$ can be represented
	by a uniform substitution: choose a formula representative of the
	class assigned to each variable. If $\varphi\in L$, closure under
	US therefore makes $\varphi\approx1$ valid in $\mathbf F_L$.
	If $\varphi\notin L$, the canonical valuation $p\mapsto[p]$
	refutes $\varphi\approx1$. Hence
	\[
	\Th(\mathbf F_L)=L,
	\qquad
	\Th(\mathbf A):=\{\varphi:\mathbf A\models\varphi\approx1\}.
	\]
	For a variety $\mathcal W$, write
	$\Th(\mathcal W)=\bigcap_{\mathbf A\in\mathcal W}\Th(\mathbf A)$.
	If $\Vcl\subseteq\mathcal W\subseteq\Vmcl$, algebraic soundness
	shows that $\Th(\mathcal W)\in\lattice$.
	Conversely, Boolean equivalence turns every equation $s\approx t$
	into the equation $(s\leftrightarrow t)\approx1$.
	Birkhoff's theorem therefore implies that a variety of Boolean
	coalition algebras is determined by its theorem set. Since
	$L\subseteq\CL$, the variety $\Var(\mathbf F_L)$ contains
	$\Vcl$ and lies in the stated interval. Consequently
	\[
	L\longmapsto\Var(\mathbf F_L),
	\qquad
	\mathcal W\longmapsto\Th(\mathcal W)
	\]
	are mutually inverse, inclusion-reversing correspondences between
	$\lattice$ and $[\Vcl,\Vmcl]$.
	
	Arbitrary intersections of nonempty families of structural theorem
	sets remain structural; their joins are the deductive closures of
	their unions. The displayed formulas also account for empty
	families. Finally, the closure operator $\Cn$ is finitary: every
	derivation from additional axioms uses only finitely many of them.
	Thus each $L\in\lattice$ is the join of the finitely generated
	extensions $\Cn(\Gamma)$ for finite $\Gamma\subseteq L$.
	Each such extension is compact by finitariness; conversely,
	compactness of $L$ applied to this join provides a finite
	$\Gamma\subseteq L$ with $L=\Cn(\Gamma)$. This proves both
	algebraicity and the claimed characterization of compact elements.
\end{proof}

Since $N$ is finite, each of the three schemes Ser, IA, and Det
has only finitely many coalition-index choices; US takes care of
its formula parameters. Hence the eight basic extensions
\begin{equation}\label{eq:basic}
	B_\Sigma:=\Cn(\Sigma),\qquad
	\Sigma\subseteq\{\Ser,\IA,\Det\},
\end{equation}
are compact elements of $\lattice$. Here $\Cn(\Sigma)$ denotes
closure under MP, US, and CN after adjoining every instance of
the schemes named in $\Sigma$. We abbreviate
$S=B_{\{\Ser\}}$, $I=B_{\{\IA\}}$, and $D=B_{\{\Det\}}$.

A \emph{coatom} of $\lattice$ is a maximal proper element below
$\CL$. Under Proposition~\ref{prop:duality}, it corresponds to a
\emph{cover} of $\Vcl$ in the subvariety interval
$[\Vcl,\Vmcl]$. We will use the subdirect decomposition theorem
and J\'onsson's lemma in this setting
\cite[Corollary~3.2]{Jonsson1967}; see also
\cite{BurrisSankappanavar1981} for standard universal-algebraic
constructions.

\begin{proposition}[The eight basic logics do not form a sublattice]
	\label{prop:non-sublattice}
	If $|N|\ge2$, then
	\[
	\MCL\subsetneq
	(\MCL+\Ser)\cap(\MCL+\IA).
	\]
	In particular, the eight basic extensions form a Boolean cube under
	inclusion, but not a sublattice of $\lattice$.
\end{proposition}
\begin{proof}
	Choose distinct agents $a,b\in N$, and let
	\[
	\iota(p,q):=
	([a]p\wedge[b]q)\to[a,b](p\wedge q)
	\]
	be an instance of IA. Put $e=[\varnothing]\top$ and
	$\Box\varphi=e\to[\varnothing]\varphi$, as in
	\eqref{eq:derived-box}, and consider
	$\delta:=\iota(p,q)\vee\Box e$.
	Because $\iota(p,q)$ is an IA instance,
	$\delta\in\MCL+\IA$. Likewise, $e$ is a Ser instance and
	the derived necessitation rule yields $\Box e$, so
	$\delta\in\MCL+\Ser$.
	
	For the strict inclusion, take a frame with three states $r,u,v$.
	At $r$, let the only available full profiles have coordinates
	$(0,0)$ and $(1,1)$ for agents $a,b$, respectively, with outcomes
	$\{u\}$ and $\{v\}$. Fix a designated action for every other agent
	in both profiles, leaving all unspecified profiles unavailable.
	Give $u$ no available full profile, and give $v$ one deterministic
	self-loop. Choose a valuation with $p$ true at $u$ and false at $v$,
	and $q$ true at $v$ and false at $u$.
	At $r$, action $0$ of $a$ guarantees $p$, and action $1$ of $b$
	guarantees $q$; however, no available joint action guarantees
	$p\wedge q$. Thus $\iota(p,q)$ is false at $r$. Since the
	accessible state $u$ is a dead end, $\Box e$ is also false at
	$r$. Hence $\delta$ is not valid on this general concurrent game,
	and soundness of \MCL\ gives $\delta\notin\MCL$.
	The asserted strict inclusion follows. The Boolean order of the
	eight basic systems is witnessed by the independent pure-defect
	frames used later in Lemma~\ref{lem:defects}.
\end{proof}

\paragraph{Relative criticality.}
Let $\mathcal K\subseteq\Vmcl$ be a subvariety. We call a
nontrivial finite algebra $\mathbf A\in\Vmcl$
\emph{$\mathcal K$-critical} if $\mathbf A\notin\mathcal K$
and every proper subalgebra and every proper homomorphic image of
$\mathbf A$ belongs to $\mathcal K$. Here a proper homomorphic
image is a quotient by a nonidentity congruence. In particular,
if $\mathbf A$ is simple, all its proper homomorphic images are
trivial and hence lie in $\mathcal K$; only its proper subalgebras
require further checking. The critical witnesses constructed below
are simple. The following sufficient criterion converts such a
witness into a relative variety cover, and thus into a coatom of
the coalition-logic lattice.

\begin{lemma}[Finite relative-cover criterion]\label{lem:finite-cover}
	Let $\mathbf A\in\Vmcl\setminus\Vcl$ be a finite simple algebra.
	If every proper subalgebra of $\mathbf A$ belongs to $\Vcl$, then
	\[
	\mathcal W:=\Var(\Vcl\cup\{\mathbf A\})
	\]
	covers $\Vcl$ in the subvariety lattice of $\Vmcl$.
	Consequently,
	\begin{equation}\label{eq:coatom-formula}
		C_{\mathbf A}:=\CL\cap\Th(\mathbf A)
	\end{equation}
	is a coatom of $\lattice$.
\end{lemma}
\begin{proof}
	Since $\mathbf A\notin\Vcl$, the inclusion
	$\Vcl\subsetneq\mathcal W\subseteq\Vmcl$ is strict.
	Let $\mathcal V$ satisfy
	\[
	\Vcl\subsetneq\mathcal V\subseteq\mathcal W,
	\]
	and choose $\mathbf D\in\mathcal V\setminus\Vcl$.
	By the subdirect representation theorem, $\mathbf D$ embeds into a
	product of subdirectly irreducible homomorphic images
	$\mathbf C_j\in\mathcal V$. At least one factor, say $\mathbf C$,
	lies outside $\Vcl$; otherwise closure of $\Vcl$ under products
	and subalgebras would force $\mathbf D\in\Vcl$. Hence
	\begin{equation}\label{eq:outside-si}
		\mathbf C\in\mathcal V\setminus\Vcl,
		\qquad \mathbf C\text{ is subdirectly irreducible}.
	\end{equation}
	
	Every algebra in $\Vmcl$ has the Boolean majority term
	\[
	m(x,y,z)=(x\wedge y)\vee(x\wedge z)\vee(y\wedge z).
	\]
	Thus $\Vmcl$, and in particular $\mathcal W$, is
	congruence-distributive. By J\'onsson's lemma in its general
	subdirectly irreducible form \cite[Corollary~3.2]{Jonsson1967},
	\begin{equation}\label{eq:jonsson-step}
		\mathbf C\in\mathbf H\mathbf S\mathbf P_u
		(\Vcl\cup\{\mathbf A\}).
	\end{equation}
	Explicitly, $\mathbf C$ is a homomorphic image of a subalgebra of
	an ultraproduct
	\[
	\mathbf P=\prod_{j\to\mathcal U}\mathbf D_j,
	\qquad \mathbf D_j\in\Vcl\cup\{\mathbf A\},
	\]
	for an ultrafilter $\mathcal U$ on an index set $J$.
	This form of J\'onsson's lemma applies to an arbitrary generating
	class; it does not assume that the class is finite.
	
	Let $J_A=\{j\in J:\mathbf D_j\cong\mathbf A\}$. Because
	$\mathbf A\notin\Vcl$, every factor indexed by $J\setminus J_A$
	belongs to $\Vcl$. If $J\setminus J_A\in\mathcal U$, the
	ultraproduct $\mathbf P$ is isomorphic to an ultraproduct of
	$\Vcl$-algebras, and hence belongs to $\Vcl$;
	then closure under subalgebras and homomorphic images
	would give $\mathbf C\in\Vcl$, a contradiction. Therefore
	$J_A\in\mathcal U$. Restricting the ultrafilter to $J_A$ and
	identifying all factors there with $\mathbf A$ makes
	$\mathbf P$ an ultrapower of $\mathbf A$.
	Since $\mathbf A$ is finite, each sequence in this ultrapower is
	$\mathcal U$-equivalent to a constant sequence: the finitely many
	fibres of the sequence form a partition, exactly one member of
	which is selected by the ultrafilter. Thus $\mathbf P\cong\mathbf A$,
	and hence
	\[
	\mathbf C\in\mathbf H\mathbf S(\mathbf A).
	\]
	
	Choose a subalgebra $\mathbf B\le\mathbf A$ and a surjective
	homomorphism $h:\mathbf B\twoheadrightarrow\mathbf C$.
	If $\mathbf B<\mathbf A$, the hypothesis and closure of $\Vcl$
	under homomorphic images imply $\mathbf C\in\Vcl$, contrary to
	\eqref{eq:outside-si}. Therefore $\mathbf B=\mathbf A$.
	Simplicity of $\mathbf A$ now forces $\ker h$ to be either the
	identity or the universal congruence. The latter would make
	$\mathbf C$ trivial, which is impossible because
	$\mathbf C\notin\Vcl$. Thus $h$ is an isomorphism and
	$\mathbf C\cong\mathbf A$.
	It follows that $\mathbf A\in\mathcal V$, so
	$\mathcal W\subseteq\mathcal V$; the reverse inclusion holds by
	assumption. Hence $\mathcal V=\mathcal W$, proving that
	$\mathcal W$ covers $\Vcl$.
	Finally,
	\[
	\Th(\mathcal W)
	=\Th(\Vcl)\cap\Th(\mathbf A)
	=\CL\cap\Th(\mathbf A).
	\]
	The order-duality of Proposition~\ref{prop:duality} turns the
	relative cover into the asserted coatom.
\end{proof}

The next lemma isolates the mechanism behind our finite critical
witnesses. Cyclic Boolean expansions are closely related to the
rotational cubes studied by Cz\'edli and Nagy
\cite[Theorem~2.1 and Lemmas~3.10--3.11]{CzedliNagy2013}.
We give the short argument needed here to distinguish the use of
cyclicity from the additional use of a prime cycle length.

\begin{lemma}[Cyclic-shift criterion]\label{lem:cycle}
	Let $n\ge2$, let $W_n=\mathbb Z/n\mathbb Z$, and let
	$\mathbf A$ be a Boolean coalition algebra with Boolean reduct
	$\setP(W_n)$. Suppose that, for some coalition $C$,
	\begin{equation}\label{eq:shift}
		f_C(X)=\{s\in W_n:s+1\in X\}.
	\end{equation}
	Then $\mathbf A$ is simple. If $n$ is prime, its only possible
	proper subalgebra is $\{\varnothing,W_n\}$. If, in addition,
	$f_D(0)=0$ and $f_D(1)=1$ for every coalition $D$, then this
	two-element set is a subalgebra and is a \CL-algebra.
\end{lemma}
\begin{proof}
	Let $\Theta$ be a nonidentity congruence of $\mathbf A$ and
	$J=\{X\subseteq W_n:X\equiv_\Theta\varnothing\}$ its Boolean
	zero-class. Since $J$ is a nonzero ideal of the finite Boolean
	algebra $\setP(W_n)$, it contains an atom $\{s\}$.
	The operation $f_C$ fixes $\varnothing$ and cyclically permutes the
	atoms. Compatibility with $\Theta$ therefore puts every cyclic
	translate of $\{s\}$ in $J$. Their join is $W_n$, so $J$
	contains the unit and $\Theta$ is the universal congruence.
	Thus $\mathbf A$ is simple for every $n\ge2$.
	
	Now let $n=p$ be prime, and let $\mathbf B\le\mathbf A$.
	Its Boolean reduct is a Boolean subalgebra of $\setP(W_p)$
	closed under the shift $g=f_C$. Since $g^p$ is the identity,
	$g$ restricts to an automorphism of this subalgebra.
	The atoms of its Boolean reduct partition $W_p$, and the cyclic
	shift acts transitively on the blocks: any point can be shifted
	to any other point, and the shift maps blocks to blocks.
	Consequently all blocks have the same cardinality, and their
	number divides $p$. There is therefore either a single block,
	giving $\{\varnothing,W_p\}$, or $p$ singleton blocks, giving
	the whole of $\setP(W_p)$. This proves the assertion about proper
	subalgebras. If every $f_D$ fixes $0$ and $1$, the two-element
	Boolean subalgebra is closed under all coalition operations,
	which restrict to the identity. It satisfies all the defining
	\CL-equations.
\end{proof}

\section{Serial coatoms and finite critical coalition algebras}\label{sec:finite-coatoms}

We now examine the serial part of the upper boundary of the
coalition-logic lattice $\lattice$, namely the coatoms belonging to
the principal interval $[S,\CL]$. The finite relative-cover criterion
of Lemma~\ref{lem:finite-cover} will turn explicit game algebras into
maximal proper intermediate logics. The constructions distinguish the
strategic condition that fails from the algebraic mechanism that
ensures criticality.

For a finite general concurrent game $G$ with state set $W$, its
\emph{complex coalition algebra} has Boolean reduct $\setP(W)$ and
operations
\begin{equation}\label{eq:complex}
	f_C(X)=\{s\in W:\text{some available $C$-profile at $s$}
	\text{ ensures }X\},\qquad X\subseteq W.
\end{equation}
By the semantic clause \eqref{eq:semantics}, these algebras belong to
$\Vmcl$. We construct three families on the cyclic state space
$W_p=\mathbb Z/p\mathbb Z$, with $p$ prime. Fix distinct agents
$a,b\in N$. All other agents use the designated dummy action $0$:
more precisely, although the game has a common action alphabet,
full profiles with a nonzero action at any other coordinate are
unavailable. We therefore specify an available full profile solely
by its pair of $a$- and $b$-coordinates. Every unlisted profile has
empty outcome set; all displayed outcomes are nonempty, and state
arithmetic is modulo $p$.

Each family is serial and has a singleton-coalition operation equal
to the cyclic shift on $\setP(W_p)$. The remaining profiles are
chosen to realize one of the following strategic patterns:
\begin{center}
	\small
	\begin{tabular}{@{}lccc l@{}}
		\toprule
		Critical game & Ser & IA & Det & Local strategic defect \\
		\midrule
		$G_p^{\mathrm{IA}}$ & yes & no & yes & incompatible available choices \\
		$G_p^{\mathrm{Det}}$ & yes & yes & no & nondeterministic full profile \\
		$G_p^{\mathrm{Mix}}$ & yes & no & no & both defects \\
		\bottomrule
	\end{tabular}
\end{center}
Here ``yes'' means that the corresponding axiom scheme is valid on
the entire frame, and ``no'' means that it has a counterinstance.
In particular, failures of IA and Det will be verified by valuations,
not merely inferred from the action-level conditions.

\subsection{Incompatible actions with deterministic outcomes}

\begin{construction}[Independence defect]\label{con:ia}
	For every prime $p\ge2$, let the common action alphabet be
	$U=\{0,1,2\}$. At each state $s\in W_p$, allow exactly the
	following pairs:
	\begin{equation}\label{eq:ia-table}
		\begin{array}{c|cccccc}
			(a,b)&(0,1)&(1,0)&(1,1)&(1,2)&(2,1)&(2,2)\\\hline
			O_s(a,b)&\{s+1\}&\{s\}&\{s+1\}&\{s\}&\{s\}&\{s+1\}.
		\end{array}
	\end{equation}
	Write $G_p^{\mathrm{IA}}$ for this game and
	$\mathbf A_p^{\mathrm{IA}}$ for its complex coalition algebra.
\end{construction}

Every state has available full profiles, all with singleton outcome
sets, so Ser and Det hold. Agent $a$ can choose $0$ to guarantee
$s+1$, while every other available action of $a$ has $s+1$ among
its aggregate outcomes. It follows that
\begin{equation}\label{eq:ia-shift}
	f_a(X)=\{s\in W_p:s+1\in X\}.
\end{equation}
To refute IA, fix a state $s$ and a valuation with
$V(p)=\{s+1\}$ and $V(q)=\{s\}$. Agent $a$ can guarantee $p$
by choosing $0$, and agent $b$ can guarantee $q$ by choosing $0$.
These choices cannot be combined into an available joint action;
indeed, no available full profile can guarantee $p\wedge q$,
because $V(p)\cap V(q)=\varnothing$. Thus
$([a]p\wedge[b]q)\to[a,b](p\wedge q)$ fails at $s$.

\begin{proposition}\label{prop:ia-family}
	For every prime $p\ge2$,
	\[
	C_p^{\mathrm{IA}}=\CL\cap\Th(\mathbf A_p^{\mathrm{IA}})
	\]
	is a coatom of $\lattice$ containing Ser and Det but not IA.
\end{proposition}
\begin{proof}
	By \eqref{eq:ia-shift} and Lemma~\ref{lem:cycle}, the finite
	algebra $\mathbf A_p^{\mathrm{IA}}$ is simple and its only proper
	subalgebra is the two-element Boolean algebra. Seriality gives
	$f_C(0)=0$ and $f_C(1)=1$ for every coalition $C$, so that
	subalgebra belongs to $\Vcl$. The displayed valuation shows that
	$\mathbf A_p^{\mathrm{IA}}\notin\Vcl$. Lemma~\ref{lem:finite-cover}
	therefore gives a coatom. Validity of Ser and Det in the game,
	and failure of IA, establish its stated type.
\end{proof}

\subsection{Independent actions with nondeterministic outcomes}

\begin{construction}[Determinism defect]\label{con:det}
	For every prime $p\ge3$, take $U=\{0,1\}$ and make all four
	pairs of actions of $a,b$ available, with outcomes
	\begin{equation}\label{eq:det-table}
		\begin{array}{c|cccc}
			(a,b)&(0,0)&(1,0)&(0,1)&(1,1)\\\hline
			O_s(a,b)&\{s+1\}&\{s+1\}&\{s+1\}&\{s,s+2\}.
		\end{array}
	\end{equation}
	Denote the game and its complex algebra by $G_p^{\mathrm{Det}}$
	and $\mathbf A_p^{\mathrm{Det}}$, respectively.
\end{construction}

The availability of every $a,b$-pair makes independent available
choices compatible, including those involving dummy agents. Thus
IA holds, and Ser is immediate. Agent $b$ can choose $0$ to
guarantee $s+1$, whereas action $1$ admits $s+1$ among its
aggregate outcomes. Hence
\begin{equation}\label{eq:det-shift}
	f_b(X)=\{s\in W_p:s+1\in X\}.
\end{equation}
For the failure of Det, choose $V(p)=\{s\}$ and
$V(q)=\{s+2\}$. Since $p\ge3$, these are distinct, and the
profile $(1,1)$ guarantees $p\vee q$. No available full profile
guarantees either $p$ or $q$ alone: the other three profiles lead
to $s+1$, and $(1,1)$ has both $s$ and $s+2$ as outcomes.
Consequently the instance
$[N](p\vee q)\to([N]p\vee[N]q)$ is false at $s$.

\begin{proposition}\label{prop:det-family}
	For every prime $p\ge3$,
	\[
	C_p^{\mathrm{Det}}=\CL\cap\Th(\mathbf A_p^{\mathrm{Det}})
	\]
	is a coatom of $\lattice$ containing Ser and IA but not Det.
\end{proposition}
\begin{proof}
	Equation~\eqref{eq:det-shift} establishes the cyclic-shift
	hypothesis of Lemma~\ref{lem:cycle}. As the game is serial, the
	only proper subalgebra is a two-element $\CL$-algebra. The
	countervaluation above shows that
	$\mathbf A_p^{\mathrm{Det}}\notin\Vcl$. Lemma~\ref{lem:finite-cover}
	provides the asserted coatom, and the verified game conditions
	determine its type.
\end{proof}

\subsection{Simultaneous independence and determinism defects}

\begin{construction}[Mixed defect]\label{con:mix}
	For each prime $p\ge5$, take $U=\{0,1,2,3\}$ and retain the
	six available pairs from \eqref{eq:ia-table}. In addition, allow
	\begin{equation}\label{eq:mix-add}
		(3,1)\longmapsto\{s+1\},\qquad
		(3,3)\longmapsto\{s+2,s+3\}.
	\end{equation}
	Write $G_p^{\mathrm{Mix}}$ and $\mathbf A_p^{\mathrm{Mix}}$
	for the resulting game and its complex algebra.
\end{construction}

Again, choosing action $0$ allows agent $a$ to guarantee $s+1$.
Every other available action of $a$, including action $3$, has
$s+1$ among its aggregate outcomes. Therefore
\begin{equation}\label{eq:mix-shift}
	f_a(X)=\{s\in W_p:s+1\in X\}.
\end{equation}
Seriality is immediate. The same valuation
$V(p)=\{s+1\}$, $V(q)=\{s\}$ as in the independence-defect
game refutes IA, since the incompatible choices $a=0$ and $b=0$
remain available separately, while no full profile can guarantee
$p\wedge q$. For Det, use a different valuation with
$V(p)=\{s+2\}$ and $V(q)=\{s+3\}$. The profile $(3,3)$
guarantees $p\vee q$, but no available full profile guarantees
either disjunct: $(3,3)$ has both outcomes, and all remaining
profiles lead to $s$ or $s+1$. The four states
$s,s+1,s+2,s+3$ are distinct because $p\ge5$.
Hence both IA and Det fail on $G_p^{\mathrm{Mix}}$.

\begin{proposition}\label{prop:mix-family}
	For every prime $p\ge5$,
	\[
	C_p^{\mathrm{Mix}}=\CL\cap\Th(\mathbf A_p^{\mathrm{Mix}})
	\]
	is a coatom of $\lattice$ containing Ser but neither IA nor Det.
\end{proposition}
\begin{proof}
	Equation~\eqref{eq:mix-shift}, together with seriality, shows
	that the finite algebra is simple and has only the two-element
	$\CL$-algebra as a proper subalgebra. It lies outside $\Vcl$
	by either countervaluation above. Apply
	Lemma~\ref{lem:finite-cover} and then the verified scheme
	validities.
\end{proof}

\subsection{Term descriptions and distinctness}

The strategic presentations have a common cyclic Boolean core.
Write
\[
g(X)=\{s\in W_p:s+1\in X\}
\]
and $g^kX$ for the $k$-fold iterate of $g$ on $X\subseteq W_p$.
The operations of the three complex algebras are expressed by the
following Boolean terms in $X$ and its cyclic translates.
\begin{table}[htbp]
	\centering
	\small
	\setlength{\tabcolsep}{5pt}
	\caption{Coalition operations as terms in the cyclic Boolean algebra.}
	\label{tab:operations}
	\begin{tabular}{@{}lccc@{}}
		\toprule
		Operation & IA family & Det family & Mix family\\
		\midrule
		$f_\varnothing(X)$ & $X\cap gX$ & $X\cap gX\cap g^2X$
		& $X\cap gX\cap g^2X\cap g^3X$\\[2pt]
		$f_a(X)$ & $gX$ & $gX$ & $gX$\\[2pt]
		$f_b(X)$ & $X$ & $gX$ & $X\cup(g^2X\cap g^3X)$\\[2pt]
		$f_{\{a,b\}}(X)$ & $X\cup gX$ & $gX\cup(X\cap g^2X)$
		& $X\cup gX\cup(g^2X\cap g^3X)$\\
		\bottomrule
	\end{tabular}
\end{table}

\begin{proposition}\label{prop:term-description}
	For each family and every prime in its stated range,
	Table~\ref{tab:operations} gives its complex-algebra operations.
	For additional dummy agents, $f_C=f_{C\cap\{a,b\}}$ for all
	$C\subseteq N$. Consequently, each algebra is term-equivalent
	to its Boolean reduct expanded by the cyclic shift $g$.
\end{proposition}
\begin{proof}
	For the empty coalition there is only one partial profile, so
	$f_\varnothing(X)$ holds at $s$ precisely when $X$ contains all
	outcomes of all available full profiles at $s$. These total outcome
	sets are, respectively, $\{s,s+1\}$, $\{s,s+1,s+2\}$, and
	$\{s,s+1,s+2,s+3\}$, giving the first row of the table.
	
	The formulas $f_a(X)=gX$ in all three families follow from
	\eqref{eq:ia-shift}, \eqref{eq:det-shift} (with the symmetric
	calculation for $a$), and \eqref{eq:mix-shift}. In the IA family,
	$b=0$ guarantees $s$, and any other available $b$-action has $s$
	among its outcomes; thus $f_b(X)=X$. In the Det family, either
	agent can guarantee $s+1$, giving $f_b(X)=gX$. In the Mix family,
	$b=0$ guarantees $s$, while $b=3$ has the unique available
	completion $(3,3)$ with outcome set $\{s+2,s+3\}$.
	The other available actions of $b$ introduce no additional
	possibilities; hence $f_b(X)=X\cup(g^2X\cap g^3X)$.
	
	For the grand coalition $\{a,b\}$, each available pair is a
	complete choice for the two active agents. In the IA family it can
	guarantee either $s$ or $s+1$; in the Det family, either $s+1$
	or the set $\{s,s+2\}$; and in the Mix family, either $s$,
	$s+1$, or the set $\{s+2,s+3\}$. This yields the final row.
	Since all available full profiles fix every dummy coordinate at
	$0$, adjoining or removing dummy agents from a coalition does not
	change its enforceable outcome sets. Thus
	$f_C=f_{C\cap\{a,b\}}$. Every coalition operation is therefore
	a term in the Boolean operations and $g$, while $g=f_a$ in all
	three families. The two signatures are term-equivalent.
\end{proof}

The cyclic Boolean algebra with shift is a known rotational
structure; see \cite{CzedliNagy2013}. What is specific to the
present construction is the choice of coalition operations and their
realization by finite games, producing three different patterns of
relative covers of $\Vcl$.

\begin{lemma}[Distinctness of coatoms with finite critical witnesses]\label{lem:distinct}
	Let $\mathbf A,\mathbf B\in\Vmcl\setminus\Vcl$ be finite simple
	algebras whose proper subalgebras belong to $\Vcl$. If
	\[
	\CL\cap\Th(\mathbf A)=\CL\cap\Th(\mathbf B),
	\]
	then $\mathbf A\cong\mathbf B$.
\end{lemma}
\begin{proof}
	By the order duality of Proposition~\ref{prop:duality}, equality
	of the two logic theories implies
	\[
	\Var(\Vcl\cup\{\mathbf A\})
	=\Var(\Vcl\cup\{\mathbf B\}).
	\]
	The simple algebra $\mathbf A$ is subdirectly irreducible.
	J\'onsson's lemma therefore gives
	$\mathbf A\in\mathbf H\mathbf S\mathbf P_u
	(\Vcl\cup\{\mathbf B\})$. The same ultrafilter argument as in
	Lemma~\ref{lem:finite-cover} excludes the $\Vcl$ factors and,
	since $\mathbf B$ is finite, yields
	$\mathbf A\in\mathbf H\mathbf S(\mathbf B)$.
	A proper subalgebra of $\mathbf B$ and any homomorphic image
	of it belong to $\Vcl$, whereas $\mathbf A\notin\Vcl$.
	Thus $\mathbf A$ must be a homomorphic image of $\mathbf B$
	itself. Simplicity of $\mathbf B$ and nontriviality of
	$\mathbf A$ force this homomorphism to be an isomorphism.
\end{proof}

\begin{theorem}[Three infinite families of serial coatoms]
	\label{thm:families}
	Suppose $|N|\ge2$. Each of the following three classes of coatoms
	of $\lattice$ contains a countably infinite family:
	\begin{enumerate}[label=\textup{(\roman*)},leftmargin=2.6em,itemsep=2pt]
		\item coatoms containing Ser and Det, but not IA;
		\item coatoms containing Ser and IA, but not Det;
		\item coatoms containing Ser, but neither IA nor Det.
	\end{enumerate}
	The three classes exhaust the possible strategic-condition types of
	coatoms in $[S,\CL]$, but the exhibited families are not claimed to
	exhaust all coatoms of any type. Every displayed coatom has a finite
	critical coalition-algebra witness; the existence of such a witness
	does not, by itself, imply finite axiomatizability of the coatom.
\end{theorem}
\begin{proof}
	Propositions~\ref{prop:ia-family}, \ref{prop:det-family}, and
	\ref{prop:mix-family} provide a prime-indexed family in each class.
	For distinct primes $p,q$, the corresponding critical algebras
	have Boolean reducts of cardinalities $2^p$ and $2^q$, hence are
	not isomorphic. Lemma~\ref{lem:distinct} shows that the associated
	coatoms are distinct. Coatoms from different listed classes are
	separated by the validity of IA and Det.
	
	Finally, every coatom in $[S,\CL]$ contains Ser. If it also
	contained both IA and Det, it would equal $\CL$, contrary to being
	a coatom. It therefore belongs to exactly one of the three listed
	classes. This gives the asserted trichotomy without implying an
	exhaustive classification within any class.
\end{proof}

\section{The nonserial boundary: an explicit splitting coatom}\label{sec:splitting}

Section~\ref{sec:finite-coatoms} provides infinitely many coatoms in
$[S,\CL]$, all containing seriality. The complementary part of the
upper boundary of the coalition-logic lattice is much more rigid: a
single coatom accounts for every intermediate logic that omits Ser.
Recall from \eqref{eq:derived-box} that
$e=[\varnothing]\top$ and $\Box\varphi=e\to[\varnothing]\varphi$.
Put
\[
B:=\MCL+\IA+\Det,
\qquad C_{\mathrm{ser}}:=B+\Box e.
\]

\begin{theorem}[Unique nonserial coatom]\label{thm:serial-splitting}
	The logic $C_{\mathrm{ser}}$ is a coatom of $\lattice$. More precisely,
	for every $L\in\lattice$,
	\begin{equation}\label{eq:split}
		L\subseteq C_{\mathrm{ser}}
		\quad\Longleftrightarrow\quad e\notin L.
	\end{equation}
	Consequently, $C_{\mathrm{ser}}$ is the unique coatom omitting Ser,
	and
	\[
	\Coat(\lattice)
	=\{C_{\mathrm{ser}}\}\ \dot\cup\ \Coat([S,\CL]).
	\]
\end{theorem}
\begin{proof}
	Write $T:=B+\Box e$. The proof has three steps.
	
	\emph{Step 1: relativization to serial states.}
	We show by induction on a $\CL$-derivation that
	\begin{equation}\label{eq:relative}
		\text{if }\varphi\in\CL,\text{ then }T\vdash e\to\varphi.
	\end{equation}
	Instances of the axioms of $\MCL$, IA, and Det belong to $B$,
	so their implications from $e$ are theorems of $T$.
	For Ser, $e\to[C]\top$ follows from MC. The claim is preserved
	under MP by propositional reasoning and under US because $e$ has
	no propositional variables. Suppose next that the last step of a
	$\CL$-derivation is an application of CN to a previously derived
	$\varphi$. By the induction hypothesis, $T\vdash e\to\varphi$.
	Derived necessitation and K for $\Box$ yield
	\[
	T\vdash\Box(e\to\varphi),\qquad
	T\vdash\Box e,\qquad
	T\vdash\Box\varphi.
	\]
	As $\Box\varphi=e\to[\varnothing]\varphi$ and
	$\MCL\vdash[C]\psi\to e$, we obtain
	$T\vdash[C]\psi\to[\varnothing]\varphi$. In particular,
	$T\vdash e\to([C]\psi\to[\varnothing]\varphi)$, which
	completes the CN step and proves \eqref{eq:relative}.
	
	\emph{Step 2: a dead-end representation of $T$.}
	Since $e$ is a Ser instance and $\CL$ is closed under derived
	$\Box$-necessitation, $\Box e\in\CL$. Together with
	$B\subseteq\CL$, this gives $T\subseteq\CL$.
	Let $F_{\mathrm{dead}}$ be the one-state frame with no available
	full profile. Every $[C]\varphi$ is false at its unique state,
	whereas $\Box e$ is true. Hence this frame validates $B$ and
	$\Box e$, but refutes $e$. Fixed-frame soundness shows that
	$T\subseteq\Log(F_{\mathrm{dead}})$ and $e\notin T$.
	We claim that
	\begin{equation}\label{eq:Tdead}
		T=\CL\cap\Log(F_{\mathrm{dead}}).
	\end{equation}
	The left-to-right inclusion is established. Conversely, let
	$\varphi\in\CL\cap\Log(F_{\mathrm{dead}})$. By
	\eqref{eq:relative}, $T\vdash e\to\varphi$.
	At any state where $e$ is false, no full action is available, so
	all formulas $[C]\psi$ are false. A structural induction shows
	that the truth of $\varphi$ at such a state agrees with its truth
	at the dead-end state under the same atomic assignment.
	Therefore $\neg e\to\varphi$ is valid on every general
	concurrent game model and hence is an $\MCL$-theorem by
	completeness \cite{LiJu2025}. Combining this theorem with
	$e\to\varphi$ proves $\varphi\in T$ and establishes
	\eqref{eq:Tdead}.
	
	\emph{Step 3: splitting and maximality.}
	Let $L\in\lattice$ and suppose $L\nsubseteq T$.
	Since $L\subseteq\CL$, choose $\theta\in L\setminus T$.
	By \eqref{eq:Tdead}, some valuation on $F_{\mathrm{dead}}$
	falsifies $\theta$. Define a uniform substitution $\sigma$
	on all propositional variables occurring in $\theta$: replace
	each variable by $\top$ or $\bot$ according to that valuation.
	The resulting variable-free formula $\theta^\sigma$ is false at
	every dead-end state, as is seen by induction on formulas.
	Consequently $\theta^\sigma\to e$ is valid on all general
	concurrent game models: at states satisfying $e$ its consequent
	holds, and at states satisfying $\neg e$ its antecedent fails.
	By $\MCL$-completeness,
	\[
	\MCL\vdash\theta^\sigma\to e.
	\]
	Since $L$ is closed under US and MP, $\theta\in L$ implies
	$e\in L$. Conversely, $e\notin T$ entails that $e\in L$
	is incompatible with $L\subseteq T$. This proves
	\eqref{eq:split}.
	
	Finally, $T\ne\CL$ because $e\notin T$, while
	$T+e=B+e=\CL$, since $e$ is equivalent over $\MCL$ to Ser.
	Every proper extension of $T$ within $\CL$ contains $e$ by
	the splitting property, and hence equals $\CL$. Thus $T$ is
	a coatom. If another coatom omitted Ser, \eqref{eq:split}
	would place it below $T$, forcing equality. The stated partition
	of the coatoms follows.
\end{proof}

\begin{remark}
	The pair $(S,C_{\mathrm{ser}})$ gives an explicit splitting of
	$\lattice$: for every $L\in\lattice$, either
	$L\subseteq C_{\mathrm{ser}}$ or $S\subseteq L$.
	This does not classify the coatoms containing Ser; the families
	in Theorem~\ref{thm:families} supply infinitely many such coatoms
	without asserting that these are exhaustive.
\end{remark}

\section{Compact meets in the coalition-logic lattice}\label{sec:meets}

We now turn from the upper boundary of $\lattice$ to its meet
operation. Although the eight basic extensions $B_\Sigma$ form a
Boolean cube under inclusion, their meets in the full lattice need
not remain in that cube. We determine exactly when such a meet is
compact, or equivalently finitely axiomatizable over $\MCL$.
Throughout this section,
$\Sigma,\Theta\subseteq\{\Ser,\IA,\Det\}$, and $B_\Sigma$
is as defined in \eqref{eq:basic}.

\begin{lemma}[Pure-defect gadgets]\label{lem:defects}
	For each $P\in\{\Ser,\IA,\Det\}$ there is a finite rooted
	general concurrent game frame $H_P$ with the following properties:
	\begin{enumerate}[label=\textup{(\roman*)},leftmargin=2.5em,itemsep=2pt]
		\item A fixed instance $\beta_P$ of the scheme $P$ can be
		falsified at the root by a suitable valuation.
		\item Both other strategic schemes are valid throughout $H_P$.
		\item Every proper successor-generated subframe of $H_P$
		validates all three schemes.
		\item Replacing the local game at the root by a single deterministic
		self-loop repairs the defect, without changing the local games
		at any other state or the atomic valuation.
	\end{enumerate}
	For $P\in\{\IA,\Det\}$, the root has an available full profile.
\end{lemma}
\begin{proof}
	Choose distinct agents $a,b\in N$; any other agents use dummy
	actions. For $P=\Ser$, take a one-state frame whose root has no
	available full profile, and put $\beta_{\Ser}=e$.
	For $P=\IA$, let the root have precisely two available deterministic
	profiles,
	\[
	(0,0)\longmapsto\{u\},\qquad
	(1,1)\longmapsto\{v\},
	\]
	where $u\ne v$ have classical deterministic self-loops. The formula
	\[
	\beta_{\IA}:=
	([a]p\wedge[b]q)\to[a,b](p\wedge q)
	\]
	is false at the root when $V(p)=\{u\}$ and $V(q)=\{v\}$.
	All local games are serial and deterministic, so the other two
	schemes are valid. For $P=\Det$, let the root have exactly one
	available full profile with outcome set $\{u,v\}$, again with
	classical self-loops at $u,v$. The instance
	\[
	\beta_{\Det}:=
	[N](p\vee q)\to([N]p\vee[N]q)
	\]
	is false under the same valuation; seriality and independence
	hold everywhere. In each construction, every generated subframe
	not containing the root consists of classical states. Replacing
	the offending root game by one deterministic self-loop makes that
	game classical and preserves everything stipulated in (iv).
\end{proof}

\paragraph{Full-action-history unfoldings.}
To compare models up to a prescribed modal depth, we record the
history unfolding of a general concurrent game model, in the form
used for tree-like games in \cite{ChenJu2026}. Given a pointed model
$(\mathcal M,r)$, a \emph{history} is a finite sequence
\[
h=(s_0,\pi_0,s_1,\ldots,\pi_{k-1},s_k),
\qquad s_0=r,\quad s_{i+1}\in O_{s_i}(\pi_i).
\]
Set $|h|=k$ and $\last(h)=s_k$, and let
$\Hist(\mathcal M,r)$ include all such histories, including $(r)$.
Write $h^\frown(\pi,t)$ for the extension of $h$ by $\pi,t$.
The unfolded model $\mathcal M^*$ has these histories as states,
uses the original action alphabet, and has outcome functions and
valuation
\begin{equation}\label{eq:unfolding}
	O_h^*(\pi)
	=\{h^\frown(\pi,t):t\in O_{\last(h)}(\pi)\},
	\qquad
	V^*(p)=\{h:\last(h)\in V(p)\}.
\end{equation}
The partial-profile outcomes are again the unions of the outcomes
of full-profile extensions. Consequently, for every coalition $C$
and partial profile $\sigma$,
\begin{equation}\label{eq:projection}
	\last[O_h^{*,C}(\sigma)]
	=O_{\last(h)}^C(\sigma).
\end{equation}
The projection need not be injective: different full profiles can
produce distinct history copies of the same successor state.
The equality, including preservation of empty outcome sets, is
what the truth-preservation argument needs.

\begin{lemma}[Truth preservation under unfolding]\label{lem:unfolding}
	For every history $h\in\Hist(\mathcal M,r)$ and formula $\varphi$,
	\[
	\mathcal M^*,h\models\varphi
	\quad\Longleftrightarrow\quad
	\mathcal M,\last(h)\models\varphi.
	\]
\end{lemma}
\begin{proof}
	We argue by induction on $\varphi$. The atomic and Boolean cases
	follow from the definition of $V^*$ and the induction hypothesis.
	For $[C]\psi$, \eqref{eq:projection} preserves availability of
	each partial profile. By the induction hypothesis, all histories
	in its unfolded outcome set satisfy $\psi$ exactly when all their
	projected outcomes satisfy $\psi$. Thus the same partial profile
	witnesses $[C]\psi$ on either side, establishing the modal case.
\end{proof}

\begin{lemma}[Depth-locality of full-action unfoldings]
	\label{lem:depth}
	Let $(\mathcal M,r)$ and $(\mathcal M',r')$ be pointed general
	concurrent models with action alphabets $U,U'$, and let $n\ge1$.
	Suppose that there is a length- and prefix-preserving bijection
	$\eta$ between the two sets of histories of length at most $n$,
	sending $(r)$ to $(r')$. Assume further that:
	\begin{enumerate}[label=\textup{(\roman*)},leftmargin=2.5em,itemsep=2pt]
		\item Matched histories of length at most $n$ agree on all atoms.
		\item There are bijections $\rho_a:U\to U'$ for all $a\in N$,
		inducing a coordinatewise bijection
		$\rho:U^N\to(U')^N$, such that for each $|h|<n$ and
		full profile $\pi$,
		\[
		\eta[O_h^*(\pi)]
		=O_{\eta(h)}^{\prime *}(\rho(\pi)).
		\]
	\end{enumerate}
	Then $\mathcal M,r$ and $\mathcal M',r'$ agree on every formula
	of modal depth at most $n$.
\end{lemma}
\begin{proof}
	We prove, simultaneously for all matched histories $h,\eta(h)$,
	agreement on formulas of depth at most $n-|h|$, by structural
	induction on formulas. Atomic and Boolean cases are immediate.
	The derived $\Box\psi$ is an abbreviation for
	$e\to[\varnothing]\psi$, and its expanded modal depth is
	$1+\operatorname{md}(\psi)$, so it requires no separate case.
	
	For $[C]\psi$ at a history $h$ with $|h|<n$, let
	$\rho_C:U^C\to(U')^C$ be the restriction of the coordinatewise
	action renaming. Taking unions of the full-profile outcome
	identities in (ii) gives
	\[
	\eta[O_h^{*,C}(\sigma)]
	=O_{\eta(h)}^{\prime *,C}(\rho_C(\sigma))
	\qquad(\sigma\in U^C).
	\]
	Since $\eta$ is a bijection of the truncated history sets, this
	identity preserves nonemptiness of the partial-profile outcome
	sets. At each successor history, the induction hypothesis applies
	to $\psi$ whenever
	$\operatorname{md}([C]\psi)\le n-|h|$.
	Thus $\sigma$ is an available profile ensuring $\psi$ exactly
	when $\rho_C(\sigma)$ is. Surjectivity of $\rho_C$ treats all
	profiles on the other side. Hence the two unfolded models agree
	on $[C]\psi$ at matched histories. Apply
	Lemma~\ref{lem:unfolding} at the root histories to finish.
\end{proof}

\begin{theorem}[Exact compact-meet criterion]\label{thm:compact-meet}
	For $|N|\ge2$ and all
	$\Sigma,\Theta\subseteq\{\Ser,\IA,\Det\}$,
	\begin{equation}\label{eq:criterion}
		B_\Sigma\cap B_\Theta\text{ is compact in }\lattice
		\quad\Longleftrightarrow\quad
		\Sigma\subseteq\Theta\ \text{or}\ \Theta\subseteq\Sigma.
	\end{equation}
	In particular, the meet of two incomparable basic extensions
	is not finitely axiomatizable over $\MCL$.
\end{theorem}
\begin{proof}
	The pure-defect frames from Lemma~\ref{lem:defects} show that
	none of the three strategic schemes follows from the other two.
	Indeed, if $P\in\Sigma\setminus\Theta$, the frame $H_P$
	validates $B_\Theta$ but refutes an instance of $P$.
	Together with the evident forward implication, this proves
	\begin{equation}\label{eq:basic-order}
		B_\Sigma\subseteq B_\Theta
		\quad\Longleftrightarrow\quad\Sigma\subseteq\Theta.
	\end{equation}
	For comparable $\Sigma,\Theta$, the meet is the smaller basic
	extension, which is compact by Proposition~\ref{prop:duality}.
	
	Now let $\Sigma,\Theta$ be incomparable. By interchanging their
	roles if necessary, we may choose
	\begin{equation}\label{eq:PQ}
		P\in\Sigma\setminus\Theta,
		\qquad Q\in\Theta\setminus\Sigma,
		\qquad Q\ne\Ser.
	\end{equation}
	Such a choice is always possible: the two nonempty, disjoint
	set differences cannot both consist solely of Ser. The condition
	$Q\ne\Ser$ ensures that the defect at the root has an available
	profile and hence a path to the later defect.
	
	For each $n\ge1$, construct a finite frame $F_n$ with common
	action alphabet $U=\{0,1\}$ and two distinct designated agents
	$a,b$. Write $\pi^{ij}$ for the full profile with coordinates
	$i,j$ at $a,b$ and coordinate $0$ at every other agent. Take
	pairwise distinct states
	\[
	x_0=r,\ x_1,\ldots,x_n=z,\ y,
	\]
	and, if $P\ne\Ser$, two further states $u,v$. Unless explicitly
	listed, a full profile is unavailable. At the root set
	\[
	\begin{array}{c|l}
		Q=\IA & O_r(\pi^{00})=\{x_1\},\quad
		O_r(\pi^{11})=\{y\}\\[2pt]
		Q=\Det& O_r(\pi^{00})=\{x_1,y\}.
	\end{array}
	\]
	At every intermediate state $x_i$ ($1\le i<n$), make only
	$\pi^{00}$ available, with outcome $\{x_{i+1}\}$. At $z$
	install the pure $P$-defect:
	\[
	\begin{array}{c|l}
		P=\Ser & \text{no available full profile}\\[2pt]
		P=\IA  & O_z(\pi^{00})=\{u\},\quad
		O_z(\pi^{11})=\{v\}\\[2pt]
		P=\Det & O_z(\pi^{00})=\{u,v\}.
	\end{array}
	\]
	At $y$ and at $u,v$ when present, only $\pi^{00}$ is available
	and yields the current state. This also covers $n=1$, when
	$x_1=z$ and no intermediate states occur. The state $z$ is
	reachable from $r$ by a path of length exactly $n$, and no
	nonroot state can reach $r$.
	
	The root has only a $Q$-defect, while $z$ has only a $P$-defect;
	all other local games are classical. Every nonroot generated
	subframe satisfies the conditions in $\Theta$, since
	$P\notin\Theta$. Soundness on those generated subframes and
	truth preservation under restriction imply
	\begin{equation}\label{eq:nonroot}
		F_n,s\models B_\Theta
		\qquad\text{for every }s\ne r,
		\text{under every valuation}.
	\end{equation}
	
	Choose instances $\beta_P,\beta_Q$ from
	Lemma~\ref{lem:defects}, renaming their propositional variables
	apart. (The Ser instance $\beta_{\Ser}=e$ has no variables.)
	A valuation can therefore falsify $\beta_Q$ at $r$ and
	$\beta_P$ at $z$ simultaneously. Define
	\begin{equation}\label{eq:gamma}
		\gamma_n:=\beta_Q\vee\Box^n\beta_P.
	\end{equation}
	Because $\beta_Q\in B_\Theta$ and $\beta_P\in B_\Sigma$,
	while derived $\Box$-necessitation is available in both logics,
	\[
	\gamma_n\in B_\Sigma\cap B_\Theta.
	\]
	However $\gamma_n$ is false at $r$ under the chosen valuation:
	$\beta_Q$ fails there, and $z$ is an $n$-step successor at
	which $\beta_P$ fails. Thus
	\[
	F_n\not\models\gamma_n.
	\]
	
	Let $F_n^+$ be obtained by replacing the local game at $z$ with
	the single deterministic self-loop
	$O_z(\pi^{00})=\{z\}$, leaving every other local game unchanged.
	The root still has only the $Q$-defect, and
	$Q\notin\Sigma$; all other states are now classical. Thus
	$F_n^+\models B_\Sigma$ by the frame conditions and
	Lemma~\ref{lem:frame-soundness}.
	Equip $F_n$ and $F_n^+$ with an arbitrary common valuation.
	Their full-action histories from $r$ coincide through length $n$:
	the only change is at $z$, which first occurs at depth $n$.
	The identity on those histories satisfies
	Lemma~\ref{lem:depth}, with identity action renamings.
	Consequently, for every $\theta\in B_\Sigma\cap B_\Theta$
	with $\operatorname{md}(\theta)<n$, validity on $F_n^+$ gives
	truth of $\theta$ at $r$ in $F_n$. Together with
	\eqref{eq:nonroot}, this establishes
	\begin{equation}\label{eq:low-depth}
		F_n\models
		\{\theta\in B_\Sigma\cap B_\Theta:
		\operatorname{md}(\theta)<n\}.
	\end{equation}
	
	Suppose, toward a contradiction, that
	$M:=B_\Sigma\cap B_\Theta$ is compact. By
	Proposition~\ref{prop:duality}, write $M=\Cn(\Gamma)$ for
	some finite $\Gamma\subseteq M$. Choose $n\ge1$ so that
	$\operatorname{md}(\varphi)<n$ for every $\varphi\in\Gamma$;
	this also covers $\Gamma=\varnothing$.
	Equation~\eqref{eq:low-depth} gives $F_n\models\Gamma$,
	so fixed-frame soundness yields
	$F_n\models\Cn(\Gamma)=M$. This contradicts
	$\gamma_n\in M$ and $F_n\not\models\gamma_n$.
	Therefore $M$ is not compact.
\end{proof}

\begin{corollary}[No uniformly bounded-depth axiomatization]
	\label{cor:unbounded-depth}
	If $\Sigma,\Theta$ are incomparable and
	$M=B_\Sigma\cap B_\Theta$, then for every integer $d\ge0$,
	\[
	\Cn\bigl(\{\theta\in M:\operatorname{md}(\theta)\le d\}\bigr)
	\subsetneq M.
	\]
	Consequently, every set of axioms generating $M$, even an
	infinite one, has unbounded modal depth.
\end{corollary}
\begin{proof}
	Choose $n>d$ in the construction for
	Theorem~\ref{thm:compact-meet}. By \eqref{eq:low-depth},
	$F_n$ validates every theorem of $M$ of depth at most $d$.
	By Lemma~\ref{lem:frame-soundness}, it also validates the
	$\Cn$-closure of those theorems. But $F_n$ refutes
	$\gamma_n\in M$, so that closure is a proper subset of $M$.
	If $M=\Cn(\Gamma)$ with all formula depths in $\Gamma$
	bounded by $d$, this strict inclusion would be impossible.
\end{proof}

\begin{lemma}[Compactness in a principal interval]
	\label{lem:interval-compact}
	Let $L$ be a complete lattice and $b$ a compact element of $L$.
	For every $x\in[b,\top]$, compactness of $x$ in $[b,\top]$
	is equivalent to compactness in $L$. If $L$ is algebraic, then
	$[b,\top]$ is algebraic as well.
\end{lemma}
\begin{proof}
	Suppose first that $x$ is compact in $L$. For a nonempty
	$Y\subseteq[b,\top]$, joins in the interval agree with those in
	$L$, so compactness of $x$ provides a finite subfamily of $Y$
	whenever $x\le\bigvee_{[b,\top]}Y$. For $Y=\varnothing$,
	the interval join is $b$; if $x\le b$, then $x=b$ and the
	empty subfamily is a witness. Thus $x$ is interval-compact.
	
	Conversely, let $x$ be compact in $[b,\top]$, and suppose
	$x\le\bigvee_L X$ for $X\subseteq L$. Because $b\le x$
	and $b$ is compact in $L$, there is a finite $X_0\subseteq X$
	with $b\le\bigvee_L X_0$. Consider
	$Y=\{b\vee t:t\in X\}\subseteq[b,\top]$.
	Its join in the interval is $b\vee\bigvee_L X$, which lies
	above $x$. Interval-compactness yields finite $X_1\subseteq X$
	such that
	\[
	x\le b\vee\bigvee_L X_1.
	\]
	It follows that
	$x\le\bigvee_L(X_0\cup X_1)$, as required.
	The reasoning also applies when $X$ is empty, using the ordinary
	empty-join convention.
	
	Finally, if $L$ is algebraic and $x\ge b$, express $x$ as the
	join of compact elements $c\le x$ in $L$. Then
	\[
	x=\bigvee\{b\vee c:c\le x,\ c\text{ compact in }L\}.
	\]
	A finite join of compact elements is compact in $L$, and hence
	in $[b,\top]$ by the first part. Thus the interval is algebraic.
\end{proof}

We call an algebraic lattice \emph{coherent} if its compact elements
form a bounded sublattice: equivalently, its top is compact and
finite meets of compact elements remain compact.

\begin{corollary}[Non-coherence of natural intervals]
	\label{cor:noncoherent}
	If $|N|\ge2$, then $\lattice$ and each of its principal intervals
	$[S,\CL]$, $[I,\CL]$, and $[D,\CL]$ is a non-coherent algebraic
	lattice: each contains two compact elements with noncompact meet.
\end{corollary}
\begin{proof}
	The lattice $\lattice$ is algebraic by
	Proposition~\ref{prop:duality}. Its compact elements
	$S$ and $I$ have noncompact meet by
	Theorem~\ref{thm:compact-meet}. For any of the three other
	intervals, choose the two basic extensions obtained by adding
	each of the two missing strategic schemes to its bottom element.
	They are compact and incomparable in $\lattice$, so their meet
	is noncompact there by Theorem~\ref{thm:compact-meet}.
	The bottom of each interval is compact in $\lattice$;
	Lemma~\ref{lem:interval-compact} therefore preserves both
	algebraicity and the relevant compactness properties. Hence the
	chosen pair has a noncompact meet also in the principal interval.
\end{proof}

\begin{remark}[The lattice-theoretic meaning of the depth obstruction]
	The obstruction is more than the independence of three axiom
	schemes. Each of the three failures---nonavailability,
	incompatible coalition choices, and nondeterministic outcomes---can
	be confined to a single local game while every other local game is
	classical. Two different defects can therefore be placed
	arbitrarily far apart without affecting formulas of bounded modal
	depth at the root. It is this separation that prevents finite, and
	even uniformly bounded-depth, axiomatizations of meets of
	incomparable basic extensions.
\end{remark}

\section{The lattice of coalition logics: scope and further problems}
\label{sec:scope}

The results of this paper concern the extension lattice $\lattice$
as a mathematical object, rather than only the eight basic coalition
logics determined by Ser, IA, and Det. The distinction is substantive:
the basic logics form a Boolean cube under inclusion, but their meets
need not be computed within that cube. The surrounding lattice has
infinitely many explicitly constructed coatoms, a uniquely determined
nonserial coatom, and noncompact meets of incomparable basic
extensions. The stronger modal-depth obstruction and the
non-coherence of natural principal intervals further show that the
compact elements of $\lattice$ need not be closed under meets.

Two complementary methods make these results possible. On the
algebraic side, the order duality with $[\Vcl,\Vmcl]$ converts
relative variety covers into coatoms; finite critical coalition
algebras yield three infinite families of such covers. On the
semantic side, local strategic defects can be separated by
arbitrarily many classical rounds, producing the exact compact-meet
criterion. The general tools---congruence distributivity, J\'onsson's
lemma, and techniques from modal extension lattices---have established
precedents \cite{Jonsson1967,Blok1980,Kracht1999}, as do the cyclic
Boolean structures used in the finite witnesses
\cite{CzedliNagy2013}. The coalition-specific contribution lies in
the game-theoretic realization of the critical operations and in the
resulting structure theorems for $\lattice$.

The scope of these conclusions is important. Variety duality
applies to every structural logic in the $\MCL$--$\CL$ interval,
but does not by itself provide game-frame completeness for each
such logic. Likewise, a finite critical algebra witnesses a
relative cover without automatically giving a finite axiomatization
of the corresponding coatom. These distinctions lead to several
concrete questions about the lattice of coalition logics.

\begin{problem}[Finite axiomatizability of critical coatoms]
	For which admissible prime--type pairs $(p,\kappa)$, with
	$\kappa\in\{\mathrm{IA},\mathrm{Det},\mathrm{Mix}\}$, is
	$C_p^\kappa$ compact in $\lattice$? Here admissibility means
	$p\ge2$ for $\kappa=\mathrm{IA}$, $p\ge3$ for
	$\kappa=\mathrm{Det}$, and $p\ge5$ for $\kappa=\mathrm{Mix}$.
	
	Baker's finite basis theorem \cite{Baker1977} applies to the
	finite algebra $\mathbf A_p^\kappa$: the signature is finite and
	$\Var(\mathbf A_p^\kappa)$ is congruence-distributive.
	Thus $\Var(\mathbf A_p^\kappa)$ has a finite equational basis.
	This does not settle finite axiomatizability of
	\[
	C_p^\kappa
	=\Th\bigl(\Vcl\vee\Var(\mathbf A_p^\kappa)\bigr)
	=\CL\cap\Th(\mathbf A_p^\kappa),
	\]
	which concerns the equational theory of a \emph{join} of varieties,
	or equivalently finite generation of the coatom over $\MCL$.
\end{problem}

\begin{problem}[The coatom spectrum and splitting]
	Determine the cardinality of $\Coat(\lattice)$ and of each of
	the three serial coatom classes in Theorem~\ref{thm:families}.
	The constructions give an infinite lower bound, while the countable
	language yields the upper bound $2^{\aleph_0}$. It remains unknown
	from the present results whether these classes contain further coatoms
	beyond the exhibited prime-indexed families.
	
	Which serial coatoms admit finite critical algebraic witnesses?
	Can the finite simple $\Vcl$-critical algebras realized by general
	concurrent games be classified? Within this restricted class,
	Lemma~\ref{lem:distinct} shows that equality of the associated
	relative covers is equivalent to isomorphism of their finite
	critical witnesses. Finally, which serial coatoms $C$, if any,
	admit a splitting partner $K\nsubseteq C$ such that every
	$L\in\lattice$ satisfies $L\subseteq C$ or $K\subseteq L$?
	How does the answer change when splitting is considered relative
	to $[S,\CL]$?
\end{problem}

\begin{problem}[Coherent principal intervals and compact meets]
	Characterize the compact elements $L\in\lattice$ for which the
	principal interval $[L,\CL]$ is coherent. Corollary~\ref{cor:noncoherent}
	excludes the intervals with bottoms $\MCL$, $S$, $I$, and $D$,
	but does not settle the general case. More broadly, for which
	finitely generated structural intermediate logics $L_1,L_2$
	is the meet $L_1\cap L_2$ again finitely generated?
\end{problem}

\paragraph{Reproducibility.}
The finite-game operations and the depth-separation constructions
admit direct finite checks. Appendix~\ref{app:verification}
describes a reproducible protocol and the supplementary program
\texttt{coalition\_checks.py}. These calculations serve as independent
checks of particular instances; the uniform statements for all
admissible primes and all modal depths rest on the proofs in the
main text.

\section*{Acknowledgements}

The authors used OpenAI ChatGPT (GPT-6), under their direction, for
substantive English-language editing and reorganization of the manuscript,
for suggesting and drafting clarifications to mathematical proofs, and for
drafting and refining the supplementary finite-model checking program.
The authors reviewed the AI-assisted material and retain full responsibility
for the mathematical results, the independent verification of proofs and
computational outputs, the accuracy of references, and the final manuscript.

\appendix
\section{Reproducible checks for finite witnesses}\label{app:verification}

The finite constructions in Sections~\ref{sec:finite-coatoms} and
\ref{sec:meets} admit direct computational checks. To make these
checks reproducible, we provide the supplementary Python script
\texttt{coalition\_checks.py}. It uses only the Python standard
library. From a directory containing the script, run
\begin{center}
	\small
	\verb|python3 coalition_checks.py --primes 2 3 5 7 --depths 1 2 3 4|
\end{center}
with assertions enabled (that is, without Python's \texttt{-O}
optimization option). The two arguments specify the finite cycle
sizes and the positive modal depths to be tested.

\emph{Cyclic game algebras.}
For each supplied prime $p$ and each admissible family
($p\ge2$ for IA, $p\ge3$ for Det, and $p\ge5$ for Mix), the
script reconstructs the full-profile outcome table using the two
active agents. It computes the operations $f_\varnothing$, $f_a$,
$f_b$, and $f_{\{a,b\}}$ for every subset of $W_p$ and compares
them with Table~\ref{tab:operations}. It then checks the defining
\MCL-algebra inequalities and evaluates all instances of the Ser,
IA, and Det equations over the relevant subsets of $W_p$. Finally,
it enumerates the partitions of $W_p$ invariant under the cyclic
shift, confirming the expected trivial/discrete partition pattern
for the tested prime sizes. These calculations concern the two
active agents; the extension to additional dummy agents is proved
in Section~\ref{sec:finite-coatoms}, not separately enumerated by
the script.

\emph{Depth-separation frames.}
The script enumerates the nine unordered incomparable pairs among
the eight basic extensions. For each pair and each requested positive
depth $n$, it constructs the separating frame $F_n$ and its repaired
version $F_n^+$ from Theorem~\ref{thm:compact-meet}. It checks the
local strategic conditions, the conditions on all nonroot generated
subframes, and the fact that repairing $z$ leaves the other local
games unchanged. It also compares the full-profile histories from
the root through depth $n$ and evaluates an explicit valuation
refuting $\gamma_n=\beta_Q\vee\Box^n\beta_P$ at the root of $F_n$.
The history comparison is the identity comparison for these concrete
frames; it does not enumerate arbitrary depth-bounded formulas.

For the command displayed above, the script tests nine admissible
family--prime instances and $36$ depth-separation instances (nine
incomparable pairs at four depths). A successful run ends with
\begin{center}
	\small
	\verb|PASS SUMMARY cyclic instances=9; incomparable pairs=9; depth instances=36|
\end{center}
Composite cycle sizes are skipped, and nonpositive depths are
rejected. Because the program exhaustively enumerates truth sets
and set partitions, it is intended for small finite instances.

These are finite regression checks, not machine-checked proofs of the
general theorems. In particular, the program does not establish
simplicity or the relative-cover property for every prime, verify
J\'onsson's lemma, or prove depth-locality and the compact-meet
criterion for arbitrary modal depth. Those conclusions follow from
the arguments in the main text.


\begin{thebibliography}{99}
	
	\bibitem{Pauly2002}
	M.~Pauly, A modal logic for coalitional power in games,
	\emph{J. Logic Comput.} \textbf{12} (2002), no.~1, 149--166.
	\href{https://doi.org/10.1093/logcom/12.1.149}{doi:10.1093/logcom/12.1.149}.
	
	\bibitem{LiJu2025}
	Y.~Li and F.~Ju, A minimal coalition logic,
	\emph{J. Logic Comput.} \textbf{35} (2025), no.~8, art.~exaf059.
	\href{https://doi.org/10.1093/logcom/exaf059}{doi:10.1093/logcom/exaf059}.
	
	\bibitem{LiJu2026}
	Y.~Li and F.~Ju, Completeness of coalition logics with seriality,
	independence of agents, or determinism,
	\emph{J. Log. Lang. Inf.} (2026), published online.
	\href{https://doi.org/10.1007/s10849-026-09457-z}{doi:10.1007/s10849-026-09457-z}.
	
	\bibitem{ChenJu2026}
	Z.~Chen and F.~Ju, Seven kinds of equivalent models for generalized
	coalition logics, \emph{Inform. Comput.} \textbf{309} (2026),
	art.~105419.
	\href{https://doi.org/10.1016/j.ic.2026.105419}{doi:10.1016/j.ic.2026.105419}.
	
	\bibitem{Jonsson1967}
	B.~J\'onsson, Algebras whose congruence lattices are distributive,
	\emph{Math. Scand.} \textbf{21} (1967), 110--121.
	\href{https://doi.org/10.7146/math.scand.a-10850}{doi:10.7146/math.scand.a-10850}.
	
	\bibitem{Blok1980}
	W.~J.~Blok, The lattice of modal logics: an algebraic investigation,
	\emph{J. Symbolic Logic} \textbf{45} (1980), no.~2, 221--236.
	\href{https://doi.org/10.2307/2273184}{doi:10.2307/2273184}.
	
	\bibitem{Kracht1999}
	M.~Kracht, \emph{Tools and Techniques in Modal Logic},
	Studies in Logic and the Foundations of Mathematics, vol.~142,
	Elsevier, 1999.
	
	\bibitem{BurrisSankappanavar1981}
	S.~Burris and H.~P.~Sankappanavar, \emph{A Course in Universal Algebra},
	Graduate Texts in Mathematics, vol.~78,
	Springer-Verlag, New York, 1981.
	
	\bibitem{Baker1977}
	K.~A.~Baker, Finite equational bases for finite algebras
	in a congruence-distributive equational class,
	\emph{Adv. Math.} \textbf{24} (1977), no.~3, 207--243.
	\href{https://doi.org/10.1016/0001-8708(77)90056-1}{doi:10.1016/0001-8708(77)90056-1}.
	
	\bibitem{CzedliNagy2013}
	G.~Cz\'edli and I.~V.~Nagy, Varieties of distributive rotational lattices,
	\emph{Acta Univ. Palack. Olomuc. Fac. Rerum Natur. Math.}
	\textbf{52} (2013), no.~1, 71--78.
	\href{https://dml.cz/handle/10338.dmlcz/143392}{DML-CZ permanent record}.
	
\end{thebibliography}
\end{document}